%% file: root.tex
\documentclass[11pt]{article}

\input{includes.tex}

\input{macros.tex}

\begin{document}

\input{title.tex}

\input{body.tex}

\bibliographystyle{plain}
    \bibliography{root.bib}
    
\end{document}

%% file: includes.tex
\usepackage{microtype}
\usepackage{amssymb}
\usepackage{centernot}
\usepackage{ dsfont }
\usepackage{xspace}
\usepackage{microtype}
\usepackage{mathtools}
\usepackage{esvect}
\usepackage{accents}
\usepackage{cases}
\usepackage{dsfont}
\usepackage{url}
\usepackage{enumerate}
\usepackage[T1]{fontenc}
\usepackage{lmodern}
\usepackage{bm}
\usepackage{todonotes} 
\usepackage{amsmath}
\usepackage{mathtools}

\usepackage{graphicx} 

\usepackage{booktabs} 
\usepackage{array} 
\usepackage{paralist} 
\usepackage{verbatim} 
\usepackage{amsmath}
\usepackage{amssymb}

\usepackage{centernot}
\usepackage{ dsfont }

\usepackage{mathrsfs}

\usepackage{thmtools,thm-restate}

\usepackage{amsthm}             
\usepackage[noblocks]{authblk}  

 \usepackage[hidelinks]{hyperref} 

\usepackage{microtype}
\usepackage{amssymb}
\usepackage{centernot}
\usepackage{ dsfont }
\usepackage{xspace}
\usepackage{microtype}
\usepackage{mathtools}
\usepackage{esvect}
\usepackage{accents}
\usepackage{cases}
\usepackage{dsfont}
\usepackage{url}
\usepackage{enumerate}
\usepackage[T1]{fontenc}
\usepackage{lmodern}
\usepackage{bm}
\usepackage{todonotes} 
\usepackage{amsmath}
\usepackage{hyperref}
\usepackage{mathtools}

\usepackage{graphicx} 

\usepackage{booktabs} 
\usepackage{array} 
\usepackage{paralist} 
\usepackage{verbatim} 
\usepackage{amsmath}
\usepackage{amssymb}

\usepackage{centernot}
\usepackage{ dsfont }
\usepackage{varwidth}

\usepackage{mathrsfs}

\usepackage{thmtools,thm-restate}
\usepackage{amsmath}
\usepackage{mathtools}
\usepackage{amssymb}
\usepackage{amsthm}
\usepackage{ stmaryrd }
\usepackage[T1]{fontenc}      
\usepackage[utf8]{inputenc}   
\usepackage{graphicx}
\usepackage{braket}
\usepackage{framed}
\usepackage{todonotes}
 \usepackage[ruled]{algorithm} 
\usepackage{algorithmic}
\usepackage{xspace}
\usepackage{microtype}
\usepackage{esvect}
\usepackage{cases}
\usepackage[T1]{fontenc}
\usepackage{url}
\usepackage{enumerate}

\usepackage{cleveref}

\usepackage{ dsfont }

\usepackage{accents}
\usepackage{cases}
\usepackage{url}
\usepackage{enumerate}
\usepackage[T1]{fontenc}
\usepackage{lmodern}
\usepackage{bm}
\usepackage{booktabs} 
\usepackage{array} 
\usepackage{paralist} 
\usepackage{verbatim} 
\usepackage{centernot}

\usepackage{mathrsfs}

%% file: macros.tex
\newlength{\arrow}

\newcommand*{\N}{\mathbb{N}}
\newcommand*{\R}{\mathbb{R}}

\DeclareMathOperator{\yes}{YES}

\DeclareMathOperator{\no}{NO}

\newcommand*{\Q}{\mathbb{Q}}

\newcommand{\ra}{\rightarrow}

\newcommand{\zo}{\{0,1\}}
\newcommand{\zost}{\{0,1\}^*}
\newcommand{\zinf}{\{0,1\}^\infty}

\renewcommand{\b}[1]{\b{#1}}
\renewcommand{\b}[1]{\boldsymbol{#1}}

\DeclareMathOperator{\adim}{adim}
\DeclareMathOperator{\Dim}{Dim}
\DeclareMathOperator{\aDim}{aDim}

\DeclareMathOperator{\AVG}{AVG}

\newtheorem{theorem}{Theorem}

\newtheorem{observation}[theorem]{Observation}

\newtheorem{corollary}{Corollary}
\newtheorem{proposition}{Proposition}

\newcommand{\tab}{\hspace*{5mm}}

\usepackage{environ}

\numberwithin{equation}{section}

\DeclareMathOperator{\constr}{alg}

\newtheorem*{mydef}{Definition}


%% file: title.tex
\title{\LARGE \bf
Algorithmic Dimensions via Learning Functions  \thanks{This research was supported in part by National Science Foundation research grants 1545028 and 1900716.}
}


\author[1]{Jack H. Lutz}
\author[2]{Andrei N. Migunov}

\affil[1]{Iowa State University, Ames, IA 50011 USA \\
    \texttt{lutz@iastate.edu}}
    
\affil[2]{Drake University, Des Moines, IA 50311 USA \\ 
    \texttt{andrei.migunov@drake.edu}}
 
\date{}

\maketitle





%% file: body.tex
\begin{abstract}
We characterize the algorithmic dimensions (i.e., the lower and upper asymptotic densities of information) of infinite binary sequences in terms of the inability of learning functions having an algorithmic constraint to detect patterns in them. Our pattern detection criterion is a quantitative extension of the criterion that Zaffora Blando used to characterize the algorithmically random (i.e., Martin-L\"{o}f random) sequences. Our proof uses Lutz's and Mayordomo's respective characterizations of algorithmic dimension in terms of gales and Kolmogorov complexity.
\end{abstract}

\section{Introduction}

Algorithmic dimension was first formulated as a $\Sigma^0_1$ effectivization of classical Hausdorff dimension \cite{lutz2000gales,jLutz03}.\footnote{Algorithmic dimension has also been called ``constructive dimension,'' ``effective Hausdorff dimension,'' and ``effective dimension'' by various authors.}  The algorithmic dimension $\adim(\Gamma)$ of a set $\Gamma$ of infinite binary sequences is, in fact, an upper bound of the Hausdorff dimension $\dim_H(\Gamma)$ of this set.  Since algorithmic dimension has the \emph{absolute stability} property that $\adim(\Gamma)$ is always the supremum of all $\adim(\{X\})$ for $X \in \Gamma$, it was natural to define $\dim(X) = \adim(\{X\})$ for all infinite binary sequences $X$ and to investigate algorithmic dimension entirely in terms of the dimensions $\dim(X)$ of individual sequences $X$.  Mayordomo \cite{jMayo02} proved that the dimension $\dim(X)$, originally defined in terms of algorithmic betting strategies called \emph{gales}, can also be characterized as the lower asymptotic density of the algorithmic information content of $X$.  The more recent \emph{point-to-set principle} \cite{lutz2018algorithmic} uses relativization to give the characterization
$$\dim_H(\Gamma) = \min_{A \subseteq \N } \sup_{X \in \Gamma} \dim^A(X)$$
of classical Hausdorff dimension.  This principle has enabled several recent uses of computability theory to prove new classical theorems about Hausdorff dimension\cite{Nlut171,NluStu17,LuSt2018,Sla21,lutz2021point,downey2019algorithmic,downey2019computability,LuLu2020,lutz2021algorithmic}.\footnote{A theorem is ``classical'' here if its statement does not involve computability or related aspects of mathematical logic.  Hence ``new classical theorem'' is not an oxymoron.}

Algorithmic dimension has the same $\Sigma^0_1$ ``level of effectivization'' as algorithmic randomness (also called ``Martin-L\"{o}f randomness'' \cite{Mart66} or, simply, ``randomness'').  In fact, every algorithmically random sequence $X$ satisfies $\dim(X) = 1$, although the converse does not hold \cite{jLutz03}.

Computable learning theory, as initiated by Gold in 1967 \cite{gold1967language,osherson1986systems,jain1999systems}, has been used to shed new light on randomness notions.  Specifically, in 2008, Osherson and Weinstein \cite{osherson2008recognizing} characterized two randomness notions, called weak 1-randomness and weak 2-randomness, for sequences $X$ in terms of the inability of computable learning functions to detect patterns in $X$.  Even more compellingly, Zaffora Blando \cite{blando2021learning} recently formulated a clever variant of Osherson and Weinstein's pattern detection criteria, called uniform weak detection, and used this to give an exact characterization of algorithmic (i.e., Martin-Löf) randomness.\footnote{Zaffora Blando's paper also characterized Schnorr randomness in terms of "computably uniform weak detection," but this is not germane to our work here.}

In this paper we introduce a quantitative version of Zaffora Blando's uniform weak detection criterion, called {\it $s$-learnability}, and we use this to characterize algorithmic dimension in terms of learning functions.  Our main theorem says that, for every infinite binary sequence $X$, $\dim(X)$ is the infimum of all nonnegative real numbers s for which some learning function $s$-learns X. Our proof of this result uses methods of Osherson, Weinstein, and Zaffora Blando, together with martingale and Kolmogorov complexity techniques of Mayordomo \cite{jMayo02}. We also characterize both the classical \emph{packing dimension} $\dim_P$ \cite{sullivan1984entropy,tric1982} and the \emph{algorithmic strong dimension} $\Dim(X)$ \cite{AHLM07} of a sequence in terms of learning functions. Along the way, we show that algorithmic randomness can also be characterized by specifically \textit{polynomial-time} computable learning functions.

\section{Preliminaries}
Let $\N$ represent the natural numbers $\{0,1,2,...\}$, $\Q$ the rationals, and $\R$ the reals. We will often use the extended naturals $\N \cup \{\infty\}$ and the extended reals $\R \cup \{\infty\}$. More often, we will refer to the respectively half open and closed intervals, $[0,\infty)$ and $[0,\infty]$.

We denote by $\zost$ the set of all (finite) binary strings, and by $\zo^\infty$ the set of all infinite binary sequences, which we call the \emph{Cantor space} $\textbf{C}$. We denote the length in bits of a string or sequence $w$ by $|w|$. The \emph{empty string} is the unique string $\lambda$ with $|\lambda| = 0$. If $Z$ is an element of $\zo^\infty$ or of $\zost$, we write $Z \upharpoonright n$ for the first $n$ bits of $Z$ if $|Z| \geq n$, and the value is undefined otherwise . Note that for all $n \in \N$, and all $Z \in \zo^\infty$, $|Z \upharpoonright n| = n$. We write $w \sqsubseteq Z$ if $w$ is a \textit{prefix} of $Z$, i.e., if $w = Z \upharpoonright |w|$. We write $w \sqsubset Z$ if $w$ is a \textit{proper prefix} of $Z$, i.e., if $w$ is a prefix of $Z$ and $w \neq Z$.

For any string $w \in \zost$, the \emph{cylinder} at $w$ is $$C_w = \{Z \in \textbf{C} \text{ | } w \sqsubseteq Z\}.$$ If $A$ is a set of strings, we denote the union of cylinders at those strings by $\llbracket A \rrbracket = \cup_{w \in A} C_w$. A (Borel) probability measure on \textbf{C} is a function $\mu:\zost \ra [0,1]$ such that $\mu(\lambda) = 1$ and $\mu(w) = \mu(w0)+\mu(w1)$ for all $w \in \zost$. Intuitively, $\mu(w)$ is the probability that $Z \in C_w$ when $Z \in \textbf{C}$ is  ``chosen according to $\mu$''. In this sense, $\mu(w)$ is an abbreviation for $\mu(C_w)$. Standard methods \cite{billingsley1995probability} extend $\mu$ from cylinders to a $\sigma$-algebra $\mathcal{F}$ on \textbf{C}, so that $(\textbf{C},\mathcal{F},\mu)$ is a probability space in the classical sense.

Most of our attention is on the uniform (Lebesgue) probability measure $\lambda$ on \textbf{C} defined by $\lambda(w) = 2^{-|w|}$ for all $w \in \zost$. We rely on context to distinguish Lebesgue measure from the empty string.

\begin{mydef}
A function $f:\zost \ra [0,\infty)$ is lower semicomputable if there exists a computable function $g:\zost \times \N \ra \Q \cap [0,\infty)$ such that for all $w \in \zost, t \in \N$,$$g(w,t) \leq g(w, t+1) \leq f(w)$$ and $$\lim_{t \ra \infty} g(w,t) = f(w).$$
\end{mydef}

We will say that a set $S$ of real numbers is \emph{uniformly left computably enumerable (uniformly left c.e.)} if there exists a lower-semicomputable function whose range is $S$.

\section{Algorithmic randomness via learning functions}
The algorithmic randomness of a sequence was originally defined in \cite{Mart66} in terms of algorithmic measure theory. In this view, an algorithmically random sequence is one which belongs to every algorithmically definable measure one set. Martin-L\"{o}f shows that all algorithmically nonrandom sequences belong to one universal algorithmically measure-zero set.

\begin{mydef}\cite{Mart66}
If $\mu$ is a probability measure on \textbf{C}, then we say a set $X$ has \emph{algorithmic $\mu$-measure zero} if there exists a computable function $g: \N \times \N \ra \zost$ such that: for every $k\in \mathbb{N}$,
  $$
   X\subseteq \bigcup_{n=0}^{\infty}C_{g(k,n)}
  $$
  and
  
   $$\sum_{n=0}^{\infty}\mu(C_{g(k,n)})\leq 2^{-k}.$$
  
\end{mydef}

\begin{mydef}\cite{Mart66}
  We say a sequence $S$ is \emph{Martin-L\"{o}f $\mu$-nonrandom} if $\{S\}$ has algorithmic $\mu$-measure zero, and \emph{Martin-L\"{o}f $\mu$-random} otherwise.
\end{mydef}

When the probability measure $\mu$ is the Lebesgue measure $\lambda$ defined in section 2, we omit it from the terminology in the preceding two definitions. Randomness can also be characterized using gambling strategies called \emph{gales}. 

\begin{mydef}\label{galedef} \cite{jLutz03}
For $s \in [0,\infty)$, a \emph{$\mu$-$s$-gale} is a function $d:\zost \ra [0,\infty)$ that satisfies the condition that $$d(w)\mu(w)^s = d(w0)\mu(w0)^s + d(w1)\mu(w1)^s,$$ for every $w \in \zost$. 
\end{mydef}

Sometimes, when the probability distribution $\mu$ is clear from context, we will refer simply to $s$-gales. A \emph{martingale} is a $1$-gale. If not stated explicitly otherwise, we assume that the \emph{initial capital} of a gale is $d(\lambda) = 1$.

\begin{mydef}
A $\mu$-$s$-gale $d$ \emph{succeeds} on a set $\Gamma$ of sequences if $$\limsup_{n \ra \infty} d(X \upharpoonright n) = \infty$$ for every sequence $X \in \Gamma.$
\end{mydef}

Ville shows the following:

\begin{theorem}\label{villeorig} \cite{Ville39} Let $\lambda(E)$ denote the Lebesgue measure of a set $E \subseteq \{0,1\}^\infty$. The following are equivalent:\\
(1.)  $\lambda(E) = 0$\\
(2.) There exists a martingale $d: \{0,1\}^* \rightarrow [0,\infty)$ that succeeds on $E$. 
\end{theorem}

Schnorr effectivizes Ville's theorem as follows:

\begin{theorem}\label{schnorrthm} \cite{DBLP:journals/mst/Schnorr71} A set of sequences $\Gamma$ is Martin-L\"{o}f random if and only if there exists no lower-semicomputable martingale that succeeds on $\Gamma$.
\end{theorem}

We will revisit gales in a later section, when we discuss dimension.
\input{randLF.tex}

\section{Classical and algorithmic dimensions}
Next, we review the definitions of classical and algorithmic dimensions.

We say a set $A$ \emph{covers} a set of sequences $\Gamma$ if for every $X \in \Gamma$ there is some $w \in A$, $w \sqsubseteq X$. Let $k \in \N$, and let $\mathcal{A}_k = \{A \mid A \text{ is a prefix set and } \forall x \in A, |x| \geq k\}.$ Let $$\mathcal{A}_k(\Gamma) = \{ A \in \mathcal{A}_k \mid A \text{ covers } \Gamma\},$$ and let $$H^s_k(\Gamma) = \inf_{A \in \mathcal{A}_k(\Gamma)} \sum_{w \in A} 2^{-s|w|}.$$ Note that this infimum is taken only over sets of cylinders, not over all possible covers. For that reason, the following function is a proxy for - and is within a constant multiplicative factor of - what is know as the \emph{$s$-dimensional Hausdorff outer measure}, in which all covers are considered.

\begin{mydef} \cite{haus19}
$H^s(\Gamma) = \lim_{k \ra \infty} H^s_k(\Gamma).$
\end{mydef}

For any set $\Gamma$, there exists some $s \in [0,\infty)$ such that for every $a < s < b$, \\
\begin{enumerate}
    \item $H^a(\Gamma) = \infty$, and
    \item $H^b(\Gamma) = 0$.
\end{enumerate}

The real number $s$ is the Hausdorff dimension of $\Gamma$:

\begin{mydef} \cite{haus19}
The \emph{Hausdorff dimension} of a set $\Gamma \subseteq \zinf$ is 
$$\dim_H(\Gamma) = \inf \{s \in [0,\infty)\text{ }|\text{ } H^s(\Gamma) = 0\}.$$
\end{mydef}

Hausdorff dimension can be characterized in terms of gales:

\begin{theorem}
\cite{lutz03a} 
$$\dim_H(\Gamma) = \inf \{s \in [0,\infty) | \text{ there exists an $s$-gale that succeeds on $\Gamma$}\}.$$
\end{theorem}

Lutz also showed that by effectivizing gales\footnote{In Lutz' original proof, \emph{supergales} are used. These satisfy the gale definition with $\leq$ in place of equality. \cite{hitchcock2003gales} shows that gales suffice for the characterization of algorithmic dimension.} at various levels, one can obtain various \emph{effective} dimension notions, including the \emph{algorithmic dimension}:

\begin{mydef}
\cite{jLutz03} The \emph{algorithmic dimension} of a set $\Gamma \subseteq \zinf$ is 
$$\adim(\Gamma) = \inf \{s \in [0,\infty) | \text{ there exists}$$ $$\text{a lower semicomputable $s$-gale that succeeds on $\Gamma$}\}.$$
\end{mydef}

We also note the following theorem:

\begin{theorem}\label{galeswap} \cite{jLutz03}
Let $\mu$ be a probability measure on $\zinf$, let $s, s' \in [0,\infty)$, and let $d,d':\zost \ra [0,\infty)$. Assume that $$d(w)\mu(w)^s = d'(w)\mu(w)^{s'}$$ for all $w \in \zost$. Then, $d$ is a $\mu$-$s$-gale if and only if $d'$ is a $\mu$-$s'$-gale.
\end{theorem}

A corollary of this theorem, which we will make use of in the main proof of the next section, is the following:

\begin{proposition} 
\label{martingalevsorder}
If $d$ is an $s$-gale that succeeds on $X$ then there exists a martingale $d'$ that succeeds on $X$ against order $h(w)=2^{(1-s)|w|}$. That is, $$\limsup_{n \ra \infty} \frac{d(X \upharpoonright n)}{h(n)} = \infty.$$
\end{proposition}

Another important classical dimension notion is the packing dimension \cite{tric1982,sullivan1984entropy}. As with the Hausdorff dimension, it is easier to define it in terms of gales than in terms of its original conception via coverings. First we define a notion of \emph{strong success} for gales:

\begin{mydef}\cite{AHLM07} 
A $\mu$-$s$-gale $d$ \emph{succeeds strongly} on a set $\Gamma$ of sequences if $$\liminf_{n \ra \infty} d(X \upharpoonright n) = \infty$$ for every sequence $X \in \Gamma.$
\end{mydef}

\begin{mydef} \cite{AHLM07} The \emph{packing dimension} of a set $\Gamma \subseteq \zinf$ is
    \[\dim_P(\Gamma) = \inf \{ s \in [0,\infty) | \text{ there exists an $s$-gale that succeeds strongly on } \Gamma\}.\]
\end{mydef}

The packing dimension can be effectivized as follows:

\begin{mydef} \cite{AHLM07}
    The \emph{algorithmic packing dimension} or \emph{algorithmic strong dimension} of a set $\Gamma \subseteq \zinf$ is \[\aDim(\Gamma) = \inf \{s \in [0,\infty) |  \text{ there exists a lower semi-computable $s$-gale}\] \[\text{ that succeeds strongly on } \Gamma\}. \]
\end{mydef}

We use the above notations $\adim$ and $\aDim$ when describing the algorithmic dimensions of sets in order to distinguish these from their classical counterparts. When applied to individual sequences, we often use $\dim$ and $\Dim$, respectively, as there is no ambiguity.

\section{Hausdorff dimension via learning functions}
\input{HdimLF.tex}

\section{Algorithmic dimension via learning functions}
\input{algdimLF.tex}

%% file: randLF.tex
\begin{mydef} \cite{osherson1986systems}
A function $l:\zost \ra \{\yes,\no\}$ is called a \emph{learning function}.
\end{mydef}

$\yes$ and $\no$ are simply aliases for 1 and 0, respectively.

One can impose resource bounds on learning functions or on properties of these functions such as their average answers `along' a string $w$. $l$ may be computable (in which case we call it a computable learning function or \emph{CLF}), or $l$ may have lower-semicomputable averages at all points in $\zost$, or any number of other resource restrictions.

\begin{mydef}\label{uwd} \cite{blando2021learning}
A learning function $l$ is said to \emph{uniformly weakly detect} that a sequence $X \in \zinf$ is patterned if and only if\\
\tab 1. $l(X \upharpoonright m) = \yes$ for infinitely many $m \in \N$, and\\
\tab 2. $\lambda(\{ Y \in \zo^\infty \mid \#\{m\in \N \mid l(Y \upharpoonright m) = \yes \} \geq n \}) \leq 2^{-n}$\\\tab for all $n \in \N$.
\end{mydef}

Zaffora Blando shows that computable learning functions and uniform weak detectability characterise Martin-L\"of randomness:

\begin{theorem}\label{zbmain} \cite{blando2021learning}
A sequence $X \in \zinf$ is Martin-L\"of random if and only if there is no computable learning function that uniformly weakly detects that $X$ is patterned.
\end{theorem}

\begin{theorem} \label{LSCLFlem}
The following are equivalent:\\
\tab 1. There exists a computable learning function  that uniformly weakly detects that $X$ is patterned.\\
\tab 2. There exists a polynomial-time computable learning function that uniformly weakly detects that $X$ is patterned.
\end{theorem}

\begin{proof}
(2) $\implies$ (1) is immediate. 



To see that (1) $\implies$ (2):

Let $l$ be a computable learning function which uniformly weakly detects that $X$ is patterned. Let $M_l(w)$ be a TM which computes $l(w)$. Algorithm \ref{alg:hat-l} specifies a learning function with the following properties:

$\hat{l} = \yes$ if there exist $w' \sqsubset w$ and $t_{w'} \in \N$ such that all of the following hold:

\begin{center}
\begin{enumerate}
    \item $|w| = |w'| + t_{w'}$
    \item $l(w') = \yes$
    \item $t_{w'} = \min\{ t \text{ | } M_l(w') = \yes \text{ after } t \text{ steps}\}$,
\end{enumerate}
\end{center}

and $\hat{l} = \no$ otherwise.

If $l$ says $\yes$ infinitely often on $X$, then there are infinitely many $w', t_{w'}$ where $M_l(w') = \yes$ after exactly $t_{w'}$ time steps. Thus there are infinitely many $w$ with $|w| = |w'|+t_{w'}$ where $\hat{l}$ says $\yes$. Thus, $\hat{l}$ says $\yes$ infinitely often on $X$.

For all $Y \in \zinf$, the number of $\yes$ given by $\hat{l}$ along $Y$ remains the same as the number given by $l$, all such answers being `delayed' until later. Thus, the measure condition is satisfied. Thus, $\hat{l}$ is a learning function which uniformly weakly detects that $X$ is patterned. 


\begin{algorithm}
\caption{ $M_{\hat{l}}$}
\label{alg:hat-l}
\begin{algorithmic}[1]
\STATE Input $w$:
\FORALL{$w' \sqsubset w$}
    \STATE Run $M_l(w')$ for exactly $|w| - |w'|$ steps.
    \IF{$M_l(w')$ prints $\yes$ for the first time after exactly $|w| - |w'|$ steps}
        \RETURN $\yes$
    \ELSE
        \STATE Continue
    \ENDIF
\ENDFOR
\RETURN $\no$
\end{algorithmic}
\end{algorithm}

\noindent$M_{\hat{l}}$ always halts and terminates in time polynomial in $|w|$. Thus, $\hat{l}$ is polynomial-time computable. 
\end{proof}

As a result, one can characterize algorithmic randomness in terms of polynomial-time computable learning functions:

\begin{corollary}
    A sequence $X \in \zinf$ is Martin-L\"of random if and only if there is no polynomial-time computable learning function that uniformly weakly detects that $X$ is patterned.
\end{corollary}

We also note a useful fact about uniform weak detectability: 

\begin{observation}\label{obsunion}
If $l_1$ uniformly weakly detects that $\Gamma_1$ is patterned, and $l_2$ uniformly weakly detects that $\Gamma_2$ is patterned, then there exists a learning function $l_3$ which uniformly weakly detects that $\Gamma_1 \cup \Gamma_2$ is patterned. This transformation preserves computability.
\end{observation}
\begin{proof}
Define $l_3$ by:

$l_3(v) = \yes$ if either $l_1(v) = \yes$ or $l_2(v) = \yes$, unless for all $v' \sqsubset v$ $l_1(v) = \no$ or for all $v' \sqsubset v$ $l_2(v) = \no$. 

That is, $l_3$ says $\yes$ whenever either $l_1$ or $l_2$ would say $\yes$, except for the first time for each.

It is easy to verify that $l_3$ says $\yes$ infinitely often on every $X \in \Gamma_1 \cup \Gamma_2$ and that the measure property is satisfied. It is also easy to show that if $l_1$ and $l_2$ are computable, then $l_3$ is as well.
\end{proof}

As a result, uniform weak detectability by computable learning functions is closed under finite unions.

%% file: HdimLF.tex
Learning functions can be used to characterize the classical (Hausdorff) dimension. Once such a characterization is in place, one can impose further restrictions on the computability of learning functions and thereby use the notion of learning to characterize dimension notions at various levels of effectivity.

We refine the definition of uniform weak detectability, by adding a requirement on the frequency of $\yes$ answers, in order to characterize dimension. Recall that we identify $\yes$ with 1 and $\no$ with 0, thus the sum in the following definition is well-defined.


\begin{mydef}
    If $l$ is a learning function,  the \emph{path average of $l$ up to $w$} is denoted \[\AVG_l(w) = \frac{\sum_{i=0}^{i=|w|} l(w \upharpoonright i)}{|w|}.\]
\end{mydef}

Let $\Delta$ be a computability restriction such as `lower semi-computable', `computable', or `all' (the absence of a restriction). 

\begin{mydef}\label{slearn}
A sequence $X$ is ($\Delta$)-\emph{$s$-learnable} if and only if there exists a function $l:\zost \ra \{\yes,\no\}$ such that the following three conditions are satisfied.

\begin{enumerate}
    \item For every $w \in \zost$, the path averages $\AVG_l(w)$ are uniformly $\Delta$-computable (in $w$).
    \item For all $n \in \N$, \[\lambda(\{Y \in \zinf \mid \#\{i \mid l(Y \upharpoonright i) = \yes \} \geq n)\}) \leq 2^{-n}.\]
    \item  \[ \limsup_{n \ra \infty} \AVG_l(X \upharpoonright n) \geq 1-s.\]
\end{enumerate}
\end{mydef}

We say that a sequence $X$ is \emph{strongly $(\Delta)$-s-learnable} if it satisfies (1) and (2) above, and satisfies condition (3) with a $\liminf$ rather than a $\limsup$.





Specifically, we will say that \emph{algorithmic} $s$-learnability corresponds to $\Delta = \Sigma^0_1$, \emph{computable} $s$-learnability corresponds to $\Delta = \Delta^0_1$ and $s$-learnability as such corresponds to no resource restriction (in other words, any learning functions at all can be used). Note that the restriction on computability is not applied to the learning function itself, but to its path averages on the elements of $\zost$.

We often refer to a learning function which satisfies these properties as an \emph{$s$-learner}, and we say it $s$-\emph{learns} a sequence or set of sequences. 

This definition is in the spirit of \cite{steifer2021note}, emphasizing learnability criteria based on frequencies of $\yes$ answers, rather than restrictions on the measure conditions.

A function $l:\zost \ra \{\no,\yes\}$ is associated to a transformation $\hat{l}: \zinf \ra \{\no,\yes\}^\infty$ defined by $\hat{l}(Y) = X$, where $X[n] = l(Y \upharpoonright n)$. In this sense, each sequence $x$ is transformed into a sequence consisting of $\yes$ and $\no$ ``bits''. For any set $S$, and function $f:\zinf \ra \zinf$, $f^{-1}[S]$ is the set of all $X$ with $f(X) \in S$. Define functions $\hat{l}$ on finite strings analogously.

\begin{observation} 
For every learner $l$ and every $Y \in \zinf$, either $l$ $s$-learns every $X \in \hat{l}^{-1}[Y]$ or none of them.
\end{observation}

This is immediate from the success criterion. $s$-learning a sequence is purely a matter of the asymptotic frequency with which some learner says $\yes$ on prefixes of that sequence, assuming the measure condition is satisfied by $l$.

Closure under finite unions also holds for $s$-learnability:

\begin{observation}
If $l_1$ $s$-learns $\Gamma_1$ and $l_2$ $s$-learns $\Gamma_2$, then there exists a learning function $l_3$ which $s$-learns $\Gamma_1 \cup \Gamma_2$.
\end{observation}
\begin{proof}
The idea is the same as in Observation  \ref{obsunion}. Define $l_3$ by:

$l_3(v) = \yes$ if either $l_1(v) = \yes$ or $l_2(v) = \yes$, unless it's the first time that either $l_1$ or $l_2$ has said $\yes$ on any $v' \sqsubseteq v$.

Note that $\limsup_{n \ra \infty} \frac{\sum^n l(X \upharpoonright i)}{n} \geq (1-s)$  implies that $\limsup_{n \ra \infty} \frac{\sum^n l(X \upharpoonright i)}{n} - \frac{k}{n} \geq (1-s)$ for any fixed $k$. \end{proof}


We will now show that $s$-learning characterizes Hausdorff dimension.

Let $\mathscr{G}(\Gamma) = \{s \in (0,\infty) \mid \text{ there exists a learning function } l \text{ which $s$-learns every}$ $X \in \Gamma\}$. The following theorem is a characterization of Hausdorff dimension in terms of learning functions.

\begin{theorem}\label{hdimlf}
For all $\Gamma \subseteq \zinf$, $$\dim_H(X) = \inf \mathscr{G}(\Gamma).$$
\end{theorem}
\begin{proof}
Assume $\dim_H(\Gamma) \leq s$. Then for all $s' > s$ there exists an $s'$-gale which succeeds on $\Gamma$, and by Proposition \ref{martingalevsorder} there exists a martingale $d$ which, for every $X \in \Gamma$ succeeds on $X$ against order $h(w) = 2^{(1-s')|w|}$. That is, $d$ doubles its money asymptotically $1-s'$ share of the time. Let

\begin{equation}
\label{gale-lf}
\gamma_d(w) =
\begin{cases}
\yes &\text{ if there exists $w' \sqsubseteq w$ such that $d(w) \geq 2\cdot d(w')$ and}\\ &\text{ $\forall \hat{w}$  satisfying $w' \sqsubset \hat{w} \sqsubset w$, $\gamma_d(\hat{w}) = \no$}\\ 
\no &\text{ otherwise }
\end{cases}
\end{equation}

$\gamma_d$ says $\yes$ every time $d$ doubles its money (attains a new $2^k$ value). The set of sequences on which $\gamma_d$ says $\yes$ infinitely often is the same as the set of sequences on which $d$ wins unbounded money \emph{against $h$}, thus it has measure zero. Thus for any $X \in \Gamma$ there must be infinitely many $X \upharpoonright n$ such that $\gamma_d(X \upharpoonright n)$ has $\yes$ density at least $(1-s')$ at $X \upharpoonright n$.

Let $n \in \N$ and let $B_n = \{w \mid d(w) \geq 2^n \text{ and } \forall v \sqsubseteq w, d(v) < 2^n\}$. $B_n$ is the (prefix-)set of all $w$ at which $d$ has accumulated $2^n$ value for `the first time'. By the Kolmogorov inequality (\cite{Ville39}), $\lambda(\cup_{w \in B_n} C_w) \leq 2^{-n}$. For all $n$,
\begin{equation*}
\begin{split}
A^{\gamma_d}_n &= \{Y \mid \#\{m \mid \gamma_d(Y \upharpoonright m) = \yes\} \geq n\}\\
&\subseteq \{Y \mid \exists k\text{ } d(Y \upharpoonright k) \geq 2^n\}\\
&= \cup_{w \in B_n} C_w,
\end{split}
\end{equation*}

Thus, for all $n$, $\lambda(A^{\gamma_d}_n) \leq \lambda(\cup_{w \in B_n} C_w) \leq 2^{-n}$. Thus $\gamma_d$ $s'$-learns every $X \in \Gamma$.

In the other direction, suppose that for all $s' > s$ there is a learning function which $s'$-learns every $X \in \Gamma$. Let $l$ be such a learning function. We show $\dim_H(\Gamma) \leq s$. Let $k,r$ be any integers. 
Let \[\hat{A}_k = \{w \mid \#\yes(w) \geq r + (1-s')|w| \text{ and } \forall v \sqsubset w, v \notin \hat{A}_k\}.\]

$\hat{A}_k \in \mathcal{A}_k(\Gamma)$ because for all $w \in \hat{A}_k$, $|w| > k$, and $\hat{A}_k$ is a prefix set.

For every $n$, \[\lambda(\{Y \mid \#\yes(Y) \geq r + (1-s')n\}) \leq 2^{-r -(1-s')n}.\]

Thus, the number of strings of length exactly $k$ which can have at least $r+(1-s')k$ $\yes$ answers is at most \[\frac{2^{-r-(1-s')k}}{2^{-k}} = 2^{-r+s'k}.\]

\[H^{s'}_k(\Gamma) \leq \sum_{w \in \hat{A}_k} 2^{-s'|w|} \leq \sum_{n =k}^\infty |\hat{A}^{=n}_k| 2^{-s'n} \leq 2^{-r+s'k}2^{-s'k} =2^{-r},\] where the third inequality is due to the Kraft inequality \cite{kraft1949device,oCovTho06}, because $\hat{A}_k$ is a prefix set.

Thus 

\begin{equation*}
    \begin{aligned}
    H^{s'}(\Gamma) &= \lim_{k \ra \infty} \inf_{A \in \mathcal{A}_k(\Gamma)} \sum_{w \in A} 2^{-s|w|}\\
    &\leq \lim_{k \ra \infty} \sum_{w \in \hat{A}_k} 2^{-s|w|}\\ &= 0.
    \end{aligned}
\end{equation*}

Thus $\dim_H(\Gamma) \leq s$. 
\end{proof}

%% file: algdimLF.tex
In this section, we show that algorithmic $s$-learning characterizes algorithmic dimension.

Recall that a learning function \textit{algorithmically $s$-learns} a set $\Gamma \subseteq \{0,1\}^\infty$ if it satisfies the definition of $s$-learning with $\Delta = \Sigma^0_1$. In other words, we require that the path-averages of $l$ on all strings $w$ are uniformly lower semicomputable real numbers.


Before we discuss the proof that algorithmic $s$-learnability characterizes algorithmic dimension, we note why it is at least \textit{seemingly difficult} to directly lower semi-compute the learning function described in Theorem \ref{hdimlf}. For one, a suitable notion of lower semicomputability is hard to establish for functions mapping to $\{0,1\}$. Secondly, the learning function $\gamma_d$ in the proof of Theorem \ref{hdimlf} relies on being able to place a $\yes$ answer every time the `underlying' martingale doubles its money for the first time. That is, every time it achieves a new value $2^k$ for the first time along some sequence. Though checking whether this happens \textit{eventually} is lower semicomputable, checking whether it has happened at a particular $w$ for the \textit{first time} is not. Instead, we can simply require that path-averages be lower semicomputable. Thus, the conditions of algorithmic $s$-learnability coincide with the existence of lower semi-computable martingales succeeding against exponential orders.

For a given martingale $d$, we define the learning function $\gamma_d$ as in Theorem \ref{hdimlf}. Note that if $d$ is lower semi-computable with computable witness $\hat{d}$, then the path averages of $\gamma_d$ are lower semi-computable uniformly in $w$ via the computable witness below (Algorithm \ref{alg:path_avg}).\\

\begin{algorithm}
\caption{$\AVG_l$}
\label{alg:path_avg}
\begin{algorithmic}[1]
\STATE Input $w, k$:
\STATE Compute all the values $\hat{d}(\lambda,k),\ldots,\hat{d}(w,k)$.
\STATE Compute $m_{w,k} = \max_{w' \sqsubseteq w} \left\{ \log_2 \left(\hat{d}(w',k)\right)\right\}$
\RETURN $\frac{m_{w,k}}{|w|}$
\end{algorithmic}
\end{algorithm}

As a result, the set \[T = \{w \mid \AVG_{\gamma_d}(w) \geq 1-s\}\] is computably enumerable, and so are all of the slices \[T_m = \{w \mid \AVG_{\gamma_d}(w) \geq 1-s \text{ and } |w| = m \}.\]

Let $\mathscr{G}_{\constr}(\Gamma) = \{s \in [0,\infty) \mid $ there exists a learning function $l$ which algorithmically $s$-learns every $X \in \Gamma\}$. 

We now prove our main theorem.

\begin{theorem}\label{algchar} For all $\Gamma \subseteq \{0,1\}^\infty$,
    \[\adim(\Gamma) = \inf  \mathscr{G}_{\constr}(\Gamma) .\]
\end{theorem}
\begin{proof}

Let $s = \inf \mathscr{G}_{\constr}(\Gamma) $. We will show $\dim(\Gamma) \leq s$.

    Let $s' > s$ and let $l$ be a learning function that algorithmically $s'$-learns every $X \in \Gamma$.

    \cite{downey2010algorithmic} Theorem 13.3.4 (\cite{jMayo02}, Theorem 3.1) states that $\dim(X) = \liminf_n \frac{C(X \upharpoonright n)}{n}$ where $C$ is the \emph{plain Kolmogorov complexity}.

    Let $T = \cup T_n$ where $T_n = \{ w \text{ $\mid$ }|w| = n \text{ and } \AVG_l(w) \geq 1-s\}$. $T$ is computably enumerable due to the first condition of $s'$-learnability.

    For each $n$, $\lambda(\cup_{w \in T_n} C_w) \leq 2^{-(1-s')n}$ as a result of the measure condition of $s'$-learnability, since every sequence $Y$ with at least $n(1-s)$ $\yes$ answers at length $n$ is in the set of sequences which have at least $n(1-s)$ $\yes$ answers total, and the measure of the latter set is at most $2^{-(1-s')n}$. Thus, since each $w \in T_n$ has measure $2^{-n}$ and all are disjoint, we have $|T_n| \leq 2^{s'n}$. As a result, in plain Kolmogorov complexity terms, we only need to supply at most $s'n$ bits to identify an element of $T$ living in $T_n$ (we get $n$ `for free' as long as we supply \textit{exactly} $s'n$ bits, and use them to identify the slice $T_n$ of $T$).

    This means that for all $w \in T$, $C(w) \leq s'|w|$. Note that for every $X \in \
    \Gamma$, there are infinitely many $n$ so that $X \upharpoonright n \in T$.

    It follows, then, that for every $X \in \Gamma$, \[\dim(X) = \liminf_n \frac{C(X \upharpoonright n)}{n} \leq \frac {s'n}{n} = s',\] and thus - since $s' > s$ was arbitrary - $\dim(\Gamma) \leq s$.

    In the other direction, we show that if $\dim(\Gamma) \leq s$, then for every $s' >s$, $\Gamma$ is algorithmically $s'$-learnable.

Again, let $s' > s$ and let $d$ be a martingale which succeeds against order $2^{(1-s')n}$ on every $X \in \Gamma$. Define the learning function $\gamma_d$, as in Theorem \ref{hdimlf} ( Eq. \ref{gale-lf}), to say $\yes$ every time $d$ doubles its money (attains a new $2^k$ value) for the first time along each path.

We have established that the values $\AVG_{\gamma_d}(w)$ are uniformly lower-semicomputable when $d$ is lower-semicomputable. We also know from the proof of Theorem \ref{hdimlf} that $\gamma_d$ satisfies the measure condition and that $\limsup_{n \ra \infty} \AVG(l,X \upharpoonright n) \geq 1-s',$ for every $X \in \Gamma$. 
\end{proof}

\section{Strong algorithmic dimension via learning functions}

In this section, we characterize the packing dimension as well as the strong algorithmic dimension of a set of binary sequences in terms of strong $s$-learning and strong algorithmic $s$-learning, respectively.


For any $\Gamma \subseteq \zinf$, let $\Gamma \upharpoonright n = \{ w \mid |w| = n \text{ and $w$ is a prefix of some }v \in \Gamma\}$.

\begin{mydef} \cite{downey2010algorithmic} For any $\Gamma \subseteq \zinf$, let the \emph{upper box-counting dimension} of $\Gamma$ be \[\overline{\dim_B(\Gamma)} = \limsup_n \frac{\log |\Gamma \upharpoonright n|}{n}.\]
\end{mydef}

\begin{theorem} \label{DimKolm} \cite{AHLM07} For every $X \in \zinf$, \[\Dim(X) = \limsup_n \frac{C(X \upharpoonright n)}{n}.\]
    
\end{theorem}

Let $\mathscr{G}_{str}(\Gamma) = \{s \in [0,\infty) \mid$ there exists a learning function $f$ which strongly $s$-learns every $ X \in \Gamma\}.$

\begin{theorem} Let $\Gamma \subseteq \zinf$. Then,
   \[\dim_P(\Gamma) = \inf \mathscr{G}_{str}(\Gamma). \]
\end{theorem}
\begin{proof}

    First, we show that if, for arbitrary $s' > s$, $l$ is a learning function that strongly $s'$-succeeds on $\Gamma$, then $\dim_P(\Gamma) \leq s$. 

    The proof is much the same as \cite{AHLM07} and \cite{downey2019algorithmic} Theorem 13.11.9, except we start with a learning function instead of a gale. Assume the hypothesis. Let $\hat{s} > s'$ be arbitrary.

    Let \[T_n = \{w \mid \AVG_l(w) \geq 1-\hat{s} \text{ and } |w| = n\}\] be the set of strings of length $n$ on which $l$ achieves the requisite density. Then, for all $X \in \Gamma,$ and for all but finitely many $n \in \N$, $X \in \cup_{w \in T_n} C_w$, i.e. \[\Gamma \subseteq \cup_i \cap_{j \geq i} \llbracket T_j \rrbracket.\]

    Let $Z_n = \cap_{m \geq n} \llbracket T_m \rrbracket$. Then, $X \in Z_n$ means that after the $n$th prefix, all remaining prefixes of $X$ have density at least $1-\hat{s}$. 

    It suffices - by the countable stability of $\dim_P$ and the fact that $\dim_P \leq \overline{\dim_B}$ (see \cite{downey2010algorithmic} 13.11.3) - to show that $\overline{\dim_B}(Z_n) \leq \hat{s}$, for all n. While the original proof uses Kolmogorov's inequality, we are not dealing with a gale. We have shown in the proof of Theorem \ref{algchar} that if $l$ is an $\hat{s}$-learner, then the set $T_n$ satisfies $|T_n| \leq 2^{\hat{s}n}.$

    Then, \[\overline{\dim_B(Z_i)} = \limsup_n \frac{\log |Z_i \upharpoonright n|}{n} \leq \limsup_n \frac{\log |T_n|}{n} \leq \hat{s}.\]

    Since $\hat{s} > s' >s$ are arbitrary, $\dim_P(\Gamma) \leq \overline{\dim_B}(\Gamma) = \overline{\dim_B}(\cup Z_i) \leq s.$

    In the other direction, simply note that if $\dim_P(\Gamma) \leq s$ then for every $s' > s$ there exists a martingale $d_{s'}$ which strongly $s'$-succeeds against order $2^{(1-s')n}$ on every $X \in \Gamma$, and the frequency with which the martingale doubles its money is always eventually bounded below. Thus, a learner defined as in Eq. \ref{gale-lf} will have the desired $\yes$ density on the same sequences.
\end{proof}

Let $\mathscr{G}^{str}_{\constr}(\Gamma) = \{s \in [0,\infty) \mid $ there exists a learning function $l$ which strongly algorithmically $s$-learns every $X \in \Gamma\}$. 

\begin{theorem}Let $\Gamma \subseteq \zinf$. Then,
   \[\aDim(\Gamma) = \inf \mathscr{G}^{str}_{\constr}(\Gamma). \]
    
\end{theorem}

\begin{proof}
    Much like the gale characterization of $\aDim$ resembles the gale characterization of $\adim$ with minor changes (see \cite{AHLM07} and \cite{downey2010algorithmic}, Corollary 13.11.12), this proof closely resembles the proof of Theorem \ref{algchar}. 
    
    We replace the $\liminf$ with a $\limsup$ in order to apply Theorem \ref{DimKolm}, and we replace the observation that infinitely many $n$ satisfy $X \upharpoonright n \in T$ with ``all but finitely many.'' 

    In the other direction, we assume there exists a lower semi-computable martingale $d_{s'}$ which succeeds strongly against order $2^{(1-s')n}$ on every $X \in \Gamma$. A learning function defined in the same way as before (Eq. \ref{gale-lf}) will have the requisite density of $\yes$ answers, and will also have uniformly lower-semicomputable averages at each $w \in \zost$, as established in Section 6.
\end{proof}


